\documentclass[journal]{IEEEtran}
\usepackage{cite}
\usepackage{amsmath}
\usepackage{algorithmic}
\usepackage{url}

\usepackage{amsmath} 
\usepackage{amssymb}  
\usepackage{amsthm}
\usepackage{color}
\usepackage{psfrag}
\usepackage{pstricks}
\usepackage{pst-plot}
\usepackage{pst-circ}
\usepackage{enumitem}
\usepackage{graphicx}
\usepackage[short]{optidef}

\ifCLASSOPTIONcompsoc 
\usepackage[caption=false,font=normalsize,labelfon t=sf,textfont=sf]{subfig} \else 
\usepackage[caption=false,font=footnotesize]{subfig}
\fi

\newcommand{\fig}{Fig.\,}
\newcommand{\tab}{Tab.\,}
\newcommand{\sect}{Sec.\,}

\newcommand{\cf}{cf.\,}
\newcommand{\ie}{i.e.\,}
\newcommand{\eg}{e.g.\,}

\newcommand{\un}[1]{\underline{#1}}

\newcommand{\R}[1]{\mathbb{R}^{#1}}
\newcommand{\set}[1]{\mathbb{#1}}
\newcommand{\e}{\begin{equation}}
\newcommand{\ee}{\end{equation}}
\newcommand{\abs}[1]{\lvert#1\rvert}  
\newcommand{\trans}{^{\top}} 
\newcommand{\opt}{^{\ast}} 
\newcommand{\ind}[1]{^{(#1)}}

\newcommand{\ts}[1]{\text{#1}}

\newtheorem{prop}{Proposition}

\newcommand{\rp}{^{\ts{RP}}}
\newcommand{\ctrl}{^{\ts{C}}}
\newcommand{\da}{^{\ts{DA}}}
\newcommand{\id}{^{\ts{ID}}}

\newcommand{\reg}{^{\ts{REG}}}
\newcommand{\rsv}{^{\ts{RES}}}
\newcommand{\sys}{^{\ts{S}}}
\newcommand{\avgw}{^{\ts{S}}}
\newcommand{\base}{^{\ts{base}}}

\newcommand{\rf}{^{\ts{ref}}}
\newcommand{\regup}{^{\ts{up}}}
\newcommand{\regdn}{^{\ts{dn}}}
\newcommand{\intTsys}{\set{T}\sys(s)}
\newcommand{\intTid}{\set{T}\id(k)}
\newcommand{\intTc}{\set{T}\ctrl(l)}
\newcommand{\hor}{T^{\ts{H}}}

\renewcommand{\sect}{Section }

\begin{document}
%
\title{Ramp-Rate-Constrained Bidding of Energy and Frequency Reserves in Real Market Settings}
%
%
%

\author{Fabian~L.~M\"uller,
		Stefan Woerner,
        and~John~Lygeros,~\IEEEmembership{Fellow,~IEEE}
%
\thanks{F. L. M\"uller is with the Automatic Control Laboratory, Swiss Federal Institute of Technology, Zurich, Switzerland, and IBM Research--Zurich, Zurich, Switzerland.
        {\tt\small fmu@zurich.ibm.com}}%
\thanks{S. Woerner is with IBM Research--Zurich, Zurich, Switzerland.
        {\tt\small wor@zurich.ibm.com}}%
\thanks{J. Lygeros is with the Automatic Control Laboratory, Swiss Federal Institute of Technology, Zurich, Switzerland.
        {\tt\small jlygeros@ethz.ch}}%
}

\maketitle

\begin{abstract}
The energetic flexibility of electric energy resources can be exploited when trading on wholesale energy and ancillary service markets. This paper considers the problem of a Balance Responsible Party to maximize its profit from trading on energy markets while simultaneously offering Secondary Frequency Reserves to the System Operator. 
However, the accurate provision of regulation power can be compromised by power ramp-rate limitations of the resources providing the service. To avoid this shortcoming, we take power ramp-rate constraints into account explicitly when computing optimal adjustable energy trading policies that are robust against uncertain activation of reserves.
The approach proposed is applicable in real market settings because it models all the different timescales of the day-ahead, intra-day, and reserve markets. Finally, the effect of different market settings and energy trading policies on the amount of available Secondary Frequency Reserves is investigated. 

\end{abstract}


\section{Introduction}
\label{s:introduction}

\IEEEPARstart{B}{alance} Responsible Parties (BRPs) act as intermediaries between electric energy resources, electricity markets, and Independent System Operators (ISOs). BRPs are legal entities representing a single or a group of electric energy systems. They are responsible for trading electric energy on energy markets to satisfy the needs of their balance group and to adjust energy schedules to match actual electricity production and/or consumption as accurately as possible. 
If the balance group of a BRP comprises systems whose electric energy generation and/or consumption is flexible, the BRP can offer this flexibility to the ISO in the form of ancillary services that are traded on dedicated flexibility markets \cite{ENTSO-E2009}. 
We consider a BRP of a single flexible system that wishes to maximize the economic profit made from trading energy on day-ahead and intra-day energy markets and offering Secondary Frequency Regulation (SFR) services to the ISO. We refer to the problem of making optimal decisions on energy and reserve markets as the \emph{energy and reserve bidding problem}.

Characterizing and exploiting the energetic flexibility of various types of electric energy resources has recently attracted growing interest because the need for flexibility in power grids is expected to increase as less predictable renewable energy resources, such as wind and solar, are connected to the power grid. In particular, the provision of SFR has been studied for heating, ventilation, and air-conditioning systems \cite{Vrettos2014a,Qureshi2016,Hao2015a}, plug-in electric vehicles \cite{Vagropoulos2013,Liu2014}, and systems behaving like a generalized battery \cite{Hao2014b,Sanandaji2013}, among others. 

The energy and reserve bidding problem faced by the BRP is to decide on the amount of SFR to offer on the reserve market and to optimally trade on energy markets. These decisions are mutually dependent and constitute a trade-off problem between the revenues from offering reserves and the costs of procuring energy. The problem has recently been formulated as a stochastic optimization problem \cite{Zhang2015a,Lymperopoulos2015} and as a robust optimization problem \cite{Vrettos2014a,Vrettos2016b,Gorecki2017,Warrington2012,Warrington2013}. The robust formulations rely on the concept of affinely adjustable robust control (AARC) policies, which have first been applied to the regulation power context in \cite{Warrington2012,Warrington2013,Bertsimas2013,Zhang2014a}, building on the initial work in \cite{BenTal2004} and the references therein. 

The provision of SFR services requires the tracking of an activation signal broadcast by the ISO on the timescale of a few seconds. The ability of a system to accurately follow such signals is limited by the physical properties of the system, in particular by power, power-ramp rate, and energy constraints. While power and energy constraints are considered by most AARC approaches, all works mentioned above neglect power ramp-rate constraints. However, the crucial importance of ramping constraints has been recognized in the context of grid balancing \cite{Borsche2014} and unit commitment \cite{Morales-Espana2014,Parvania2016}, and has lead to the adoption of new performance-based schemes for the remuneration of SFR provision by several ISOs \cite{FERC2011}. But the methods developed in \cite{Borsche2014,Morales-Espana2014,Parvania2016} cannot be directly applied to the energy and reserve bidding problem considered here.


In our formulation of the bidding problem, we adopt the AARC formulation used in \cite{Vrettos2014a,Vrettos2016b,Gorecki2017,Warrington2012,Warrington2013}, and make two main contributions. 
First, in contrast to the aforementioned works, we explicitly incorporate constraints on power-ramp rates in our modeling framework to guarantee the feasibility and accurate provision of SFR. Instead of the commonly used piece-wise constant power schedules, we base our formulation on piece-wise affine and continuous power trajectories, which makes it possible to accurately model power-ramp rates on arbitrary timescales. We then derive a reformulation of the AARC optimization problem which is necessary to guarantee the feasibility of the piece-wise affine solutions.
Second, in our problem formulation we include all the different timescales and market lead times involved in the bidding problem. This allows us to study how different energy and reserve market parameters affect the amount of SFR capacity a system can provide. 
These two contributions make our approach versatile and adaptable, and pave the way for applying the energy and reserve bidding problem in real-world settings, where accurate modeling of system constraints and market specifications is indispensable. 

%

The energy and reserve bidding problem of a BRP is introduced in \sect\ref{s:problemDescription}. In \sect\ref{s:reqsPowerRef}, we illustrate how piece-wise affine power schedules can account for finite ramp rates. Section \ref{s:flexDescription} provides a robust reformulation of the the energetic flexibility of an energy resource. It is used in \sect\ref{s:scheduling} to formulate and solve the energy and reserve bidding problem. A conclusion is provided in \sect\ref{s:conclusion}.

\section{Problem Description}
\label{s:problemDescription}

We consider a BRP that participates in energy and reserve markets and whose balance group comprises a single system. The BRP strives to schedule the electric energy generation and/or consumption of the system so as to maximize the profit made from trading energy and from offering SFR services over a planning horizon of duration $\hor$. The bidding problem involves different timescales and requires the BRP to take decisions on multiple stages both offline, \ie, prior to the time of delivery, and online, \ie, at delivery.

\subsection{Offline decisions}
\textit{1) Reserve capacity market:} An ISO is responsible for the safe and reliable operation of the transmission grid. In particular, it is accountable for keeping the supply and demand of electric energy balanced at all times. To be able to compensate fluctuating generation and consumption and other unexpected disturbances in a timely manner, it procures ancillary services on dedicated markets ahead of time. Different ancillary services exist, and their specifications vary among ISOs \cite{Swissgrid2017,ENTSO-E2009,FERC_MarketPrimer}. We focus on Secondary Frequency Regulation\footnote{Secondary Frequency Regulation is also referred to as Frequency Restoration Reserves, Spinning Reserves, or Load-Frequency Control.} and adopt the product specifications that apply to the Swiss SFR market, \cf\cite{Swissgrid2017}: every week on Tuesday not later than at 1 p.m., a BRP can bid a constant and symmetric capacity ${\gamma\in\R{+}_0}$ for the subsequent week \cite{Swissgrid2017}.\footnote{More precisely, multiple bids consisting of capacity-price pairs can be submitted.} If the bid is accepted, the ISO has the right, but not the obligation, to ask the BRP to deviate from the planned power reference by at most $\pm\gamma$ units of power at any time during the subsequent week. The SFR tendering period is denoted by $T\rsv$ (1 week in Switzerland). We choose the planning horizon to comprise a single SFR tendering period, \ie, ${\hor=T\rsv=1}$ week. In return for keeping the reserve capacity available over the time $T\rsv$, the BRP receives the capacity reservation payment ${R\rsv:=c\rsv\gamma}$, where $c\rsv$ is the reserve capacity price  \cite{Swissgrid2017}. 

\textit{2) Day-ahead energy market:} On the day-ahead energy market, a BRP can trade energy for every time interval of duration $T\da$ of the next day. The Swiss day-ahead market runs on an hourly timescale, \ie, ${T\da=1}$ h, and closes at 11 a.m. on the day before delivery \cite{Abrell2016}. The outcome of the day-ahead markets is the energy schedule ${e\da\in\R{N\da}}$ of energy quantities to be produced or consumed during each of the ${N\da=7\cdot 24=168}$ time intervals in the planning horizon $\hor$. The day-ahead energy procurement costs are ${C\da:=c^{\ts{DA}\trans} e\da}$, where ${c\da\in\R{N\da}}$ denotes the day-ahead market clearing prices, which are unknown at the time of bidding.

\textit{3) Intra-day energy market:} To make adjustments to the day-ahead energy schedule during the day of delivery, the BRP can trade energy continuously on the intra-day markets for time intervals of durations ${T\id}$. 
We focus on quarter-hourly energy block products, \ie, ${T\id=15}$ min, as traded on the Swiss intra-day market with a lead time of 1 hour \cite{Abrell2016} and denote by $N\id$ the number of intra-day time slots in the planning horizon $\hor$, \ie, ${N\id=7\cdot 24\cdot 4=672}$. The energy procurement costs of the intra-day energy schedule ${e\id\in\R{N\id}}$ are ${C\id:=c^{\ts{ID}\trans} e\id}$, where ${c\id\in\R{N\id}}$ denotes the intra-day market prices, which are unknown at the time of bidding.

\subsection{Online decisions}
As will be discussed in detail in \sect\ref{s:reqsPowerRef}, the day-ahead and intra-day energy schedules of a BRP are translated into a continuous-time power reference. The BRP is responsible for ensuring that its balance group follows this power reference in order to comply with all energy contracts concluded on the markets. BRPs are held accountable for any mismatches between their actual and planned energy schedules that arise from unforeseen outages of generation units or inaccuracies of load forecasts. The ISO constantly observes the state of the power grid and, in case of imbalances, takes measures to reestablish the nominal operating state. SFR is activated to bring the grid frequency back to its nominal value and make actual energy profiles match their corresponding schedules \cite{Swissgrid2017}. This is achieved by adjusting the power consumption or generation of all energy resources that have offered SFR according to an activation signal computed and broadcast by the ISO \cite{ENTSO-E2009}.
Power deviations due to SFR activation can result in energy deviations from the energy schedule of a BRP. In Switzerland, this so-called up- and down-regulation energy is measured over each time interval $T\id$ separately for positive and negative activation, and is denoted by $e\regup$ and $e\regdn$, respectively. The BRP is remunerated for tracking the activation signal based on the amount of up- and down-regulation energy delivered, \ie, the remuneration is  ${R\reg:=c^{\ts{up}\trans} e\regup - c^{\ts{dn}\trans} e\regdn,}$
where ${c\regup,c\regdn\in\R{N\id}}$ are the corresponding regulation energy prices \cite{Swissgrid2017}.


\subsection{The energy and reserve bidding problem}
\label{ss:problemDescription}
The energy and reserve bidding problem of the BRP is to decide how much reserve capacity $\gamma$ to offer on the SFR market over the planning horizon $\hor$ and how to trade day-ahead and intra-day energy $e\da$ and $e\id$, respectively, such as to maximize the total expected profit. These decisions must be made subject to the constraints that the reference power trajectory associated with the energy schedules $e\da$ and $e\id$ \textit{i)} meets the energy needs and physical constraints of the systems in the balance group, \textit{ii)} complies with the energy contracts concluded on the energy markets, and \textit{iii)} keeps the SFR power capacity $\gamma$ offered to the ISO available at all times. The more reserve capacity the BRP offers, the more restrictive become the constraints on possible trades in the energy markets. Consequently, offering reserves and trading energy constitutes a trade-off between the reserve reward $R\rsv$, the potential remuneration of regulation energy $R\reg$, and the energy procurement costs ${C\da+C\id}$.

\subsection{Different timescales}
The bidding problem takes place on different timescales. Offline decisions are made on the timescales $T\rsv$, $T\da$, and $T\id$ corresponding to the reserve market, and the day-ahead and the intra-day energy market, respectively. Online activation of SFR occurs at the faster timescale $T\ctrl$. The dynamics of the system are discretized on an intermediate timescale $T\sys$, which depends on the time constants of the system, the frequency with which set points can be changed, and possible computational challenges of the resulting problem. Commonly, $T\sys$ will lie in the range of 5--15 min. In general, it holds that ${\hor\geq T\rsv\geq T\da\geq T\id\geq T\sys\gg T\ctrl}$. For simplicity, we assume that longer time horizons are integer multiples of shorter ones, and for ${\ast\in\{\ts{RES},\ts{DA},\ts{ID},\ts{S},\ts{C}\}}$ define ${N^{\ast}=\hor/T^{\ast}}$, ${\set{N}^{\ast}=\{1,\dots,N^{\ast}\}}$, and $\set{T}^{\ast}(k)$ the continuous time interval ${[(k-1)T^{\ast},kT^{\ast})}$.
%
%
Different timescales are used in different market regions. All timescales $T^{\ast}$ are used as parameters in our approach making it versatile and applicable to the various market settings. Here we consider the Swiss market setting and use ${\hor=T\rsv=1}$ week, ${T\da=1}$ h, ${T\id=15}$ min, ${T\sys=5}$ min, and ${T\ctrl=1}$ s.


\section{Energy Baseline and Power Reference}
\label{s:reqsPowerRef}

\subsection{Energy baseline}

The day-ahead and intra-day energy trades of a BRP result in the energy market baseline ${e\base\in\R{N\id}}$, which comprises the total energy trades in each interval of duration $T\id$ in the planning horizon. The energy baseline is defined as
\begin{equation}
	e\base_k:=e\da_j T\id/T\da+e\id_k,\ j:=\lceil k T\id/T\da\rceil,\ k\in\set{N}\id,
\end{equation}
and is written in vector form as ${e\base=Me\da+e\id}$ with ${e\base:=[e\base_1,\dots,e\base_{N\id}]\trans}$ and ${M\in\R{N\id\times N\da}}$.

\subsection{Power reference}
It has been common in recent works, \eg\cite{Vrettos2014a,Qureshi2016}, to associate with the energy baseline $e\base$ the piece-wise \emph{constant} continuous-time power reference 
\begin{equation}
p(t):= e\base_k/T\id,\ t\in\intTid,\ k\in\set{N}\id,
\label{eq:defPrefConst}
\end{equation}
that the BRP is committed to follow. Whereas the energy content of $p(t)$ over each time interval $\intTid$ matches $e\base_k$ as required by market contracts, implementing the instantaneous changes of $p(t)$ at the interface of subsequent market time intervals $\intTid$ would require infinite ramping capabilities. Inaccurate tracking of the power reference \eqref{eq:defPrefConst} can cause the grid frequency to deviate from its nominal value when switching between market time intervals, making it necessary to activate frequency reserves \cite{Borsche2014,Morales-Espana2014}.   
To alleviate this difficulty, we propose to translate the energy baseline $e\base$ into a continuous-time piece-wise \emph{affine} \cite{Morales-Espana2014} and continuous power reference $p\rf(t)$ whose breakpoint intervals are each of duration $T\sys$. The continuous-time power reference $p\rf(t)$ is fully defined by its breakpoint values ${p\rf:=[p\rf_0,\dots,p\rf_{N\sys}]\trans}$ according to
\begin{equation}
	p\rf(t):=p\rf_{s-1}+(p\rf_{s}-p\rf_{s-1})(t-(s-1)T\sys)/T\sys,
	\label{eq:defContPref}
\end{equation}
for ${t\in\intTsys}$, ${s\in\set{N}\sys}$. An example power reference is given in \fig\ref{fig:baseline}.
In accordance with ENTSO-E standards \cite{ENTSO-E2009}, the power reference comprises $N\id$ periods during which the power is constant and equal to ${e\base_k/T\id}$, ${k\in\set{N}\id}$. These periods of constant power are connected by linear ramps, each of duration $T\rp$ as illustrated in \fig\ref{fig:baseline}. We assume that $T\rp$ is an integer multiple of $T\sys$ and that ${0\leq T\rp\leq T\id}$. For the Swiss case of ${T\id=15}$ min, ${T\rp=10}$ min, and ${T\sys=5}$ min, the relation between $e\base$ and $p\rf$ is the following:
\begin{align}
p\rf_s =
\begin{cases}
	e\base_1/T\id, & s=0,\\
	e\base_k/T\id, & s=3k-2,3k-1,\\
	& k\in\set{N}\id,\\	
	(e\base_k+e\base_{k+1})/(2T\id), & s=3k,\,k\in\set{N}^{\ts{ID}-1},\\
	e\base_{N\id}/T\id, & s=N\sys.
\end{cases}
\label{eq:mapEtoPref}
\end{align} 
Thus, the power reference can be expressed compactly as
\begin{equation}
	p\rf = Re\base=RMe\da+Re\id,
	\label{eq:defPowerRef}
\end{equation}
where we encode \eqref{eq:mapEtoPref} by ${R\in\R{(N\sys+1)\times N\id}}$. 
Note that \eqref{eq:mapEtoPref} reverts to \eqref{eq:defPrefConst} in the limit case ${T\rp=0}$. For all ${T\rp>0}$, our choice  \eqref{eq:defContPref} is continuous and, thus, allows us to model the power-ramp rates. 
%
%
We make the assumption that a BRP with energy baseline $e\base$ will use the piece-wise affine continuous-time power trajectory $p\rf(t)$ \eqref{eq:defContPref} defined by its breakpoint values $p\rf$ \eqref{eq:defPowerRef} as its power reference.


%
\begin{figure}[tb]
\setlength{\abovecaptionskip}{-3pt}
\setlength{\belowcaptionskip}{-5pt}
\centering
\begin{pspicture*}(-0.6,-0.75)(8.5,4)
	
	\pspolygon[fillstyle=solid,fillcolor=lightgray,linestyle=none](0,0.5)(1,0.5)(2,3.0)(2.5,3.0)(3.5,1)(4,1)(5,2)(5.2,2)(5,1.5)(5.2,1)(5,0.5)(5.2,0)(0,0)
	\pspolygon[fillstyle=solid,fillcolor=lightgray,linestyle=none](5.3,2)(5.5,2)(6.5,1.5)(7.5,1.5)(7.5,0)(5.3,0)(5.1,0.5)(5.3,1)(5.1,1.5)
	
	\psline[linestyle=solid,linecolor=gray]{}(0,-0.6)(0,0)
	\psline[linestyle=solid,linecolor=gray]{}(1.5,-0.6)(1.5,3.2)
	\psline[linestyle=solid,linecolor=gray]{}(3.0,-0.08)(3.0,3.2)
	\psline[linestyle=solid,linecolor=gray]{}(4.5,-0.6)(4.5,3.2)
	\psline[linestyle=solid,linecolor=gray]{}(6.0,-0.6)(6.0,3.2)
	\psline[linestyle=solid,linecolor=gray]{}(7.5,-0.6)(7.5,3.2)
	
	\psline[linestyle=dotted,dotsep=1pt,linecolor=gray]{}(2,-0.6)(2,3.0)
	\psline[linestyle=dotted,dotsep=1pt,linecolor=gray]{}(2.5,-0.6)(2.5,3.0)
	\psline[linestyle=dotted,dotsep=1pt,linecolor=gray]{}(3.5,-0.6)(3.5,1.0)
	
	\psline[]{-}(0,0)(5.2,0)\psline[]{->}(5.3,0)(7.8,0)
	\psline[]{->}(0,0)(0,3.5)
	\psline[]{}(5.1,-0.1)(5.3,0.1) 
	\psline[]{}(5.2,-0.1)(5.4,0.1) 
	\rput[c](7.95,0){$t$}
	\rput[bl](-0.05,3.55){Power}
	
	\psline[]{}(0,0.5)(0.5,0.5)(1,0.5)(1.5,1.75)(2,3.0)(2.5,3.0)(3,2)(3.5,1)(4,1)(4.5,1.5)(5,2)
	\qdisk(0,0.5){2pt}
	\qdisk(0.5,0.5){2pt}
	\qdisk(1,0.5){2pt}
	\qdisk(1.5,1.75){2pt}
	\qdisk(2,3.0){2pt}
	\qdisk(2.5,3.0){2pt}
	\qdisk(3,2){2pt}
	\qdisk(3.5,1){2pt}
	\qdisk(4,1){2pt}
	\qdisk(4.5,1.5){2pt}
	\qdisk(5,2){2pt}
	\psline[]{}(5,2)(5.2,2)
	\psline[]{}(5.3,2)(5.5,2)
	\psline[]{}(5.5,2)(6,1.75)(6.5,1.5)(7,1.5)(7.5,1.5)
	\qdisk(5.5,2){2pt}
	\qdisk(6,1.75){2pt}
	\qdisk(6.5,1.5){2pt}
	\qdisk(7,1.5){2pt}
	\qdisk(7.5,1.5){2pt}
	\psline[]{}(5.1,1.9)(5.3,2.1) 
	\psline[]{}(5.2,1.9)(5.4,2.1) 
		
	\psline[]{<->}(0,-0.5)(1.5,-0.5)
	\rput[t](0.75,-0.15){$T\id$}
	
	\psline[]{<->}(2,-0.5)(2.5,-0.5)
	\rput[t](2.25,-0.15){$T\sys$}
	
	\psline[]{<->}(2.5,-0.5)(3.5,-0.5)
	\rput[t](3,-0.15){$T\rp$}
	
	\rput[b](0.7,0.05){$e_1\rf$}
	\rput[b](2.25,0.05){$e_2\rf$}
	\rput[b](3.75,0.05){$e_3\rf$}
	\rput[b](6.75,0.05){$e_{N\id}\rf$}

	\psline[]{->}(0.8,1.5)(0.8,0.5)\rput[b](0.7,1.5){$e\base_1/T\id$}
	\psline[]{->}(2.25,3.4)(2.25,3)\rput[b](2.25,3.4){$e\base_2/T\id$}
	\psline[]{->}(3.75,2.5)(3.75,1.0)\rput[b](3.75,2.5){$e\base_3/T\id$}
	\psline[]{->}(6.75,2.5)(6.75,1.5)\rput[b](6.75,2.5){$e\base_{N\id}/T\id$}
	
	\rput[b](-0.25,0.45){$p\rf_0$}
	\rput[b](0.5,0.6){$p\rf_1$}
	\rput[b](1.4,0.4){$p\rf_2$}
	\rput[b](1.9,1.6){$p\rf_3$}
	\rput[b](1.8,2.9){$p\rf_4$}
	\rput[b](2.9,2.9){$p\rf_5$}
	\rput[br](2.9,1.8){$p\rf_6$}
	\rput[br](3.4,0.8){$p\rf_7$}
	\rput[br](4.6,0.8){$p\rf_8$}
	\rput[br](4.5,1.5){$p\rf_{9}$}
	\rput[b](7.9,1.3){$p\rf_{N\sys}$}
	
\end{pspicture*}
\caption{The piece-wise affine continuous-time power reference $p\rf(t)$ (solid black line) defined by its breakpoint values ${p\rf:=[p\rf_0,\dots,p\rf_{N\sys}]\trans}$ with breakpoint intervals $T\sys$. The energy content (shaded gray area) of $p\rf(t)$ over the time interval $\intTid$ is $e_k\rf$. The power ramps of duration $T\rp$ are located symmetrically at the intersections of market time intervals.} 
\label{fig:baseline}
\end{figure}
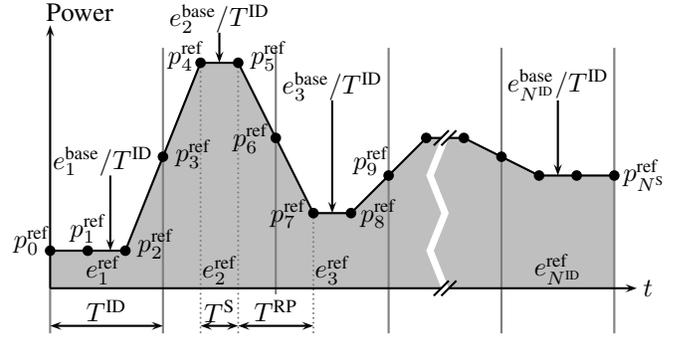

\subsection{Activation of SFR}
\label{ss:activationOfSFR}
The SFR service is activated by the ISO by sending out an activation signal to all BRPs that offer SFR for the corresponding time interval. The activation signal ${w:=[w_1,\dots,w_{N\ctrl}]\trans}$ is a discrete-time signal with ${w_l\in\set{W}:=[-1,1]}$, ${l=1,\dots,N\ctrl}$. The signal is computed by the ISO and broadcast \emph{sequentially} on the timescale $T\ctrl$ (\eg, ${T\ctrl=1-5}$ s in Continental Europe \cite{ENTSO-E2009}, ${T\ctrl=1}$ s in Switzerland \cite{Lymperopoulos2015}, and ${T\ctrl=4-6}$ s in the USA \cite{FERC2011}). That is, the activation $w_l$ becomes available to the BRP only at time $(l-1)T\ctrl$. We interpret the discrete signal $w$ as the continuous-time piece-wise affine and continuous activation signal
\begin{equation}
	w(t):= w_{l-1}+(w_l-w_{l-1})(t-(l-1)T\ctrl)/T\ctrl,
\end{equation}
where ${t\in\intTc}$, ${l\in\set{N}^{C}}$. We denote the set of all such functions $w(t)$ by $\mathcal{W}$, which is a closed set of continuous functions.
%

The target power level $p\opt(t)$ that the BRP is responsible to track continuously is a function of the continuous time $t$ and the activation $w(t)$ and is given by
\begin{equation}
	p\opt(t,w):=p\rf(t)+\gamma w(t).
	\label{eq:defTargetPower}
\end{equation}
%
%
%
Compared with the power reference $p\rf(t)$, which is piece-wise affine on the timescale $T\sys$, the activation $w(t)$ can vary at the higher rate $T\ctrl$. Thus, accurate tracking of the target power level \eqref{eq:defTargetPower} is challenging. The minimum required tracking accuracy of $p\opt(t)$ is specified by the ISOs individually (\cf\cite{TestSC} for details on the Swiss requirements). In the USA, ISOs are required to base their regulation compensation payment $R\reg$ not only on the quantity of regulation provided, but must also take into account how accurately the activation signal has been tracked \cite{FERC2011}. This motivates the use of a continuous power reference according to \eqref{eq:defPowerRef} and the explicit consideration of ramp-rate constraints. 


\subsection{Adjustable energy schedules}
\label{ss:adpatingSchedules}

The decision on the amount of SFR to be offered is made once for each reserve tendering period and cannot be changed afterwards. In contrast, energy schedules can be adjusted from day to day or intra-day by trading energy on respective markets. These adjustments can depend on past SFR activation, for example. Let ${\tilde{w}\in\set{W}^{N\avgw}}$ be the activation signal $w(t)$ averaged over time intervals of duration $T\avgw$, \ie,
\begin{equation}
	\tilde{w}_s:=\frac{1}{T\sys}\int_{(s-1)T\sys}^{T\sys} w(t) dt,\ s\in\set{N}\sys.
\end{equation}
%
%
Following the work in \cite{Warrington2012,Warrington2013,Zhang2014,Fabietti2016,Gorecki2017}, we express the adjustable day-ahead and intra-day energy schedules as affine functions of the averaged activation signal:
\begin{align}
\begin{split}
	e\da(\tilde{w}) &= Q\da A\da \tilde{w}+q\da,\\
	e\id(\tilde{w}) &= Q\id A\id \tilde{w}+q\id, 
	\label{eq:defControllers}
\end{split}
\end{align}
with parameters ${Q\da\in\R{N\da\times N\da}}$, ${q\da\in\R{N\da}}$, ${Q\id\in\R{N\id\times N\id}}$, and ${q\id\in\R{N\id}}$. The matrices ${A\da\in\R{N\da\times N\avgw}}$ and ${A\id\in\R{N\id\times N\avgw}}$ are used to average $\tilde{w}$ over time intervals of duration $T\da$ and $T\id$, respectively. For instance, $A\id$ is block-diagonal with each block given by a row-vector of length $T\id/T\sys$ whose elements all equal $T\sys/T\id$. By plugging \eqref{eq:defControllers} into \eqref{eq:defPowerRef} we can write the power reference as an affine function of the averaged activation:
\begin{equation}
	p\rf(\tilde{w}) = R(MQ\da A\da + Q\id A\id)\tilde{w}+R(Mq\da + q\id).
	\label{eq:pRefDepW}
\end{equation} 

The policies \eqref{eq:defControllers} describe how much energy should be traded on the day-ahead and intra-day energy markets depending on the average amount of SFR activated during past time intervals. The policies must be causal, \ie, they can only depend on past activation, and must respect the different lead times of the markets. These requirements can be incorporated by imposing a particular structure onto $Q\da$ and $Q\id$. Figure \ref{fig:structureQ} illustrates which elements of $Q\da$ and $Q\id$ can be nonzero. Because the day-ahead market has a lead time of 13 h, the trading decisions for the second day can only depend on the first 11 h of the first day. Similarly, the ${4\cdot 15}$ min lead time of the intra-day market forces the first 5 rows of $Q\id$ to be zero. The size of the bidding problem can be reduced by not taking into account all the available historic activation data in \eqref{eq:defControllers} but only considering the $N\da_{\ts{lb}}$ and $N\id_{\ts{lb}}$ most recent measurements. The policy parameters $N\da_{\ts{lb}}$ and $N\id_{\ts{lb}}$ determine the width of the blocks in $Q\da$ and the number of off-diagonals in the lower-triangular part of $Q\id$, respectively, \cf\fig\ref{fig:structureQ}.

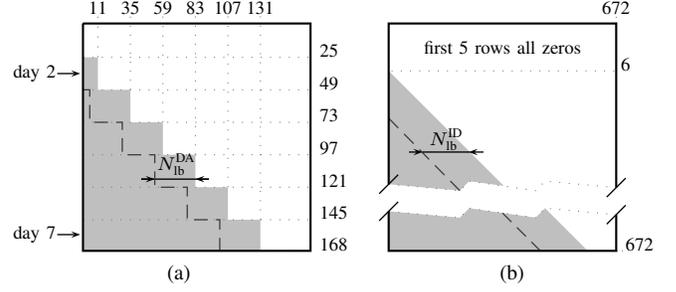
\begin{figure}[tb] 
\setlength{\belowcaptionskip}{-5pt}
\centering
\psset{unit=0.018}
\subfloat[]{
\begin{pspicture*}(-55,18)(197,-169)

	\pspolygon[fillstyle=solid,fillcolor=lightgray,linestyle=none](0,-25)(11,-25)(11,-49)(35,-49)(35,-73)(59,-73)(59,-97)(83,-97)(83,-121)(107,-121)(107,-145)(131,-145)(131,-168)(0,-168)(0,0)

	\psline[linewidth=0.5pt]{->}(-19,-37)(-2,-37)\rput[r](-20,-37){\scriptsize day 2}
	\psline[linewidth=0.5pt]{->}(-19,-156)(-2,-156)\rput[r](-20,-156){\scriptsize day 7}
	
	\psline[linecolor=gray,linestyle=dotted,linewidth=0.5pt]{-}(0,-25)(173,-25)
	\psline[linecolor=gray,linestyle=dotted,linewidth=0.5pt]{-}(0,-49)(173,-49)
	\psline[linecolor=gray,linestyle=dotted,linewidth=0.5pt]{-}(0,-73)(173,-73)
	\psline[linecolor=gray,linestyle=dotted,linewidth=0.5pt]{-}(0,-97)(173,-97)
	\psline[linecolor=gray,linestyle=dotted,linewidth=0.5pt]{-}(0,-121)(173,-121)
	\psline[linecolor=gray,linestyle=dotted,linewidth=0.5pt]{-}(0,-145)(173,-145)
	\psline[linecolor=gray,linestyle=dotted,linewidth=0.5pt]{-}(0,-168)(173,-168)
	
	\psline[linecolor=gray,linestyle=dotted,linewidth=0.5pt]{-}(11,-25)(11,5)
	\psline[linecolor=gray,linestyle=dotted,linewidth=0.5pt]{-}(35,-49)(35,5)
	\psline[linecolor=gray,linestyle=dotted,linewidth=0.5pt]{-}(59,-73)(59,5)
	\psline[linecolor=gray,linestyle=dotted,linewidth=0.5pt]{-}(83,-97)(83,5)
	\psline[linecolor=gray,linestyle=dotted,linewidth=0.5pt]{-}(107,-121)(107,5)
	\psline[linecolor=gray,linestyle=dotted,linewidth=0.5pt]{-}(131,-145)(131,5)
	
	\rput[l](175,-20){\scriptsize 25}
	\rput[l](175,-44){\scriptsize 49}
	\rput[l](175,-68){\scriptsize 73}
	\rput[l](175,-92){\scriptsize 97}
	\rput[l](175,-116){\scriptsize 121}
	\rput[l](175,-140){\scriptsize 145}
	\rput[l](175,-163){\scriptsize 168}
	
	\rput[B](11,7){\scriptsize 11}
	\rput[B](35,7){\scriptsize 35}
	\rput[B](59,7){\scriptsize 59}
	\rput[B](83,7){\scriptsize 83}
	\rput[B](107,7){\scriptsize 107}
	\rput[B](131,7){\scriptsize 131}
	
	\psline[linewidth=0.5pt]{-}(43,-115)(93,-115)	
	\psline[linewidth=0.5pt]{->}(43,-115)(53,-115)
	\psline[linewidth=0.5pt]{<-}(83,-115)(93,-115)
	\rput[bc](69,-104){\scriptsize $N\da_{\ts{lb}}$}
	
		\psline[linecolor=darkgray,linestyle=dashed,linewidth=0.5pt](0,-49)(5,-49)(5,-73)(29,-73)(29,-97)(53,-97)(53,-121)(77,-121)(77,-145)(101,-145)(101,-168)

	\psline[]{-}(0,0)(168,0)(168,-168)(0,-168)(0,0)
\end{pspicture*}
} 
\hfil
\subfloat[]{
\begin{pspicture*}(-15,18)(197,-169)
	\pspolygon[fillstyle=solid,fillcolor=lightgray,linestyle=none](0,-35)(0,-118.75)(52.5,-123.5)(59.5,-116.5)(85,-120)
	\pspolygon[fillstyle=solid,fillcolor=lightgray,linestyle=none](0,-138.75)(52.5,-143.5)(59.5,-136.5)(108.5,-143.5)(115.5,-136.5)(147,-168)(0,-168)

	\rput[c](84,-17.5){{\scriptsize first 5 rows all zeros}}

	\psline[linecolor=darkgray,linestyle=dashed,linewidth=0.5pt]{-}(0,-70)(52.5,-123.5)
	\psline[linecolor=darkgray,linestyle=dashed,linewidth=0.5pt]{-}(84,-140)(112,-168)
	
	\psline[linewidth=0.5pt]{-}(25,-95)(60,-95)
	\psline[linewidth=0.5pt]{<-}(59,-95)(70,-95)
	\psline[linewidth=0.5pt]{->}(15,-95)(25,-95)
	\rput[b](42.5,-94){\scriptsize $N\id_{\ts{lb}}$}

	\psline[linecolor=gray,linestyle=dotted,linewidth=0.5pt]{-}(168,-168)(173,-168)
	\psline[linecolor=gray,linestyle=dotted,linewidth=0.5pt]{-}(2,-35)(173,-35)
	
	\psline[linecolor=gray,linestyle=dotted,linewidth=0.5pt]{-}(168,0)(168,5)
	
	\psline[linecolor=gray,linestyle=dotted,linewidth=0.5pt]{-}(0,-118.75)(52.5,-123.5)(59.5,-116.5)(108.5,-123.5)(115.5,-116.5)(168,-121.25)
	\psline[linecolor=gray,linestyle=dotted,linewidth=0.5pt]{-}(0,-138.75)(52.5,-143.5)(59.5,-136.5)(108.5,-143.5)(115.5,-136.5)(168,-141.25)

	\rput[l](175,-163){\scriptsize 672}
	\rput[B](175,-35){\scriptsize 6}
	
	\rput[B](168,7){\scriptsize 672}
	
	\psline[]{}(0,-118.75)(0,0)(168,0)(168,-121.25)
	\psline[]{}(0,-138.75)(0,-168)(168,-168)(168,-141.25)

	\psline[linewidth=0.5pt]{}(-7,-125.75)(7,-111.75)
	\psline[linewidth=0.5pt]{}(-7,-145.75)(7,-131.75)
	\psline[linewidth=0.5pt]{}(161,-128.25)(175,-114.25)
	\psline[linewidth=0.5pt]{}(161,-148.25)(175,-134.25)

\end{pspicture*}
} 
\caption{Structure of $Q\da$ (a) and $Q\id$ (b) for the case ${N\da=168}$ and ${N\id=672}$. The shaded grey areas denote the location of potentially nonzero elements. The amount of historic activation data considered by the policies \eqref{eq:defControllers} is governed by the parameters $N\da_{\ts{lb}}$ and $N\id_{\ts{lb}}$.} 
\label{fig:structureQ} 
\end{figure}

\section{Description of flexibility}
\label{s:flexDescription}
The decisions a BRP makes on the different energy and reserve markets are restricted by the requirement that the target power trajectory \eqref{eq:defTargetPower} must be feasible for all realizations of the unknown activation signal and must satisfy all physical constraints of the system. 
We consider a system with power, power ramp-rate, and state (\eg energy) constraints. 

\subsection{Power constraints}
The power the system can draw from or feed into the power grid is limited. For all ${t\in\intTsys}$, ${s\in\set{N}\sys}$, we require that
\begin{equation}
	\un{p}_s \leq p\opt(t,w)\leq \bar{p}_s,\ \forall w\in\set{W}^{N\ctrl}, 
	\label{eq:powerConstrCt}
\end{equation}
where ${\un{p},\bar{p}\in\R{N\sys}}$ denote the piece-wise constant bounds on power. The constraints above are satisfied if and only if for all ${\tilde{w} \in\set{W}^{N\sys}}$ it holds that
\begin{align}
\begin{split}
p\rf_s(\tilde{w}) + \gamma \leq
\begin{cases}
	\bar{p}_1, & s=0,\\
	\min\{\bar{p}_{s},\,\bar{p}_{s+1}\}, & s=1,\dots,N\sys-1,\\
	\bar{p}_{N\sys}, & s=N\sys,\\
\end{cases}\\
p\rf_s(\tilde{w}) - \gamma \geq
\begin{cases}
	\un{p}_1, & s=0,\\
	\max\{\un{p}_{s},\,\un{p}_{s+1}\}, & s=1,\dots,N\sys-1,\\
	\un{p}_{N\sys}, & s=N\sys.\\
\end{cases}
\end{split}
\end{align}

The robust counterparts of the above constraints are derived according to \cite{Bertsimas2011} by computing the maximum and minimum of $p\rf(\tilde{w})$ with regard to $\tilde{w}$ according to \eqref{eq:pRefDepW}. An analytic solution exists because each element of $\tilde{w}$ is restricted to $\set{W}$: the power constraints \eqref{eq:powerConstrCt} are satisfied if and only if for all ${\tilde{w} \in\set{W}^{N\sys}}$
\begin{align}
\begin{split}
	p\rf(\tilde{w}) &\leq \abs{R(MQ\da A\da + Q\id A\id)}\mathbf{1}+R(Mq\da+q\id),\\
	p\rf(\tilde{w}) &\geq -\abs{R(MQ\da A\da + Q\id A\id)}\mathbf{1}+R(Mq\da+q\id),
\end{split}
	\label{eq:powerConstrRobust}
\end{align}
where $\mathbf{1}$ denotes a column vector of ones of appropriate size.


\subsection{Power ramp-rate constraints}
\label{ss:rampConstr}	
Limits on the rate at which power can vary over time play an important role, in particular for providing ancillary services with high accuracy. For all ${t\in\intTsys}$, ${s\in\set{N}\sys}$, we require that the rate of change of the target power be bounded, \ie,
\begin{equation}
	\un{r}_s \leq \partial p\opt(t,w)/\partial t \leq \bar{r}_s,\ \forall w\in\set{W}^{N\ctrl}, 
	\label{eq:rampConstrCt}
\end{equation}
where ${\un{r},\bar{r}\in\R{N\sys}}$ denote the piece-wise constant power ramp-rate limits. The above inequalities are satisfied if and only if for all ${\tilde{w} \in\set{W}^{N\sys}}$ it holds that
\begin{equation}
\begin{split}
	(p\rf_{s}(\tilde{w})-p\rf_{s-1}(\tilde{w}))/T\sys+2\gamma/T\ctrl &\leq \bar{r}_s,\\
	(p\rf_{s}(\tilde{w})-p\rf_{s-1}(\tilde{w}))/T\sys-2\gamma/T\ctrl &\geq \un{r}_s,\ s\in\set{N}\sys.
	\label{eq:rampConstr}
\end{split}
\end{equation}
The above inequalities illustrate that ramp-rate constraints limit the sum of the rate of change of the power reference and the rate of change of activated regulation power. They are equivalent to
\begin{align}
\begin{split}
&\hspace*{4mm}\abs{(R_{s+1}-R_{s})(MQ\da A\da + Q\id A\id)}\mathbf{1}\\
&\hspace*{2mm}+(R_{s+1}-R_s)(Mq\da+q\id)+2\gamma T\sys/T\ctrl \leq \bar{r}_s T\sys,\\
&-\abs{(R_{s+1}-R_{s})(MQ\da A\da + Q\id A\id)}\mathbf{1}\\
&\hspace*{2mm}+(R_{s+1}-R_{s})(Mq\da+q\id)-2\gamma T\sys/T\ctrl \geq \un{r}_s T\sys,
\end{split}
\label{eq:rampConstrRobust}
\end{align}
with ${s\in\set{N}\sys}$. The case of maximum negative ramp rates is illustrated in \fig\ref{fig:ramps}, where SFR activation changes from maximum positive activation (${w_{l-1}=1}$) to maximum negative activation (${w_{l}=-1}$) within one control time step $T\ctrl$. In this case, maximum negative ramp rates are required during the time interval $[(l-1)T\ctrl,lT\ctrl]$. 

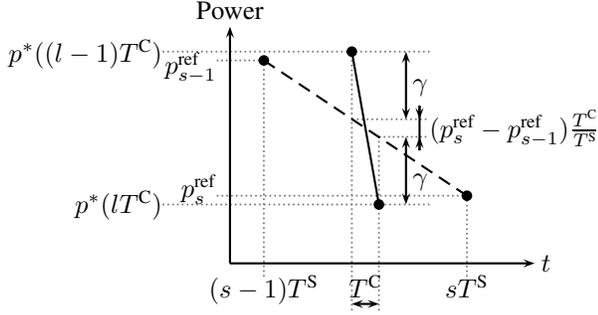
\begin{figure}[tb]
\setlength{\belowcaptionskip}{-5pt}
\centering
\psset{unit=0.9}
\begin{pspicture*}(-3.3,-0.7)(5.4,3.9)
	
	\psline[linestyle=dotted,dotsep=1pt,linecolor=gray]{}(0.5,-0.2)(0.5,3.0)
	\psline[linestyle=dotted,dotsep=1pt,linecolor=gray]{}(3.5,-0.2)(3.5,1.0)
	\psline[linestyle=dotted,dotsep=1pt,linecolor=gray]{}(1.8,-0.7)(1.8,3.13)
	\psline[linestyle=dotted,dotsep=1pt,linecolor=gray]{}(2.2,-0.7)(2.2,1.87)
	
	\psline[linestyle=dotted,dotsep=1pt,linecolor=gray]{}(-0.2,3)(0.5,3)
	\psline[linestyle=dotted,dotsep=1pt,linecolor=gray]{}(-0.2,1)(3.5,1)
	\psline[linestyle=dotted,dotsep=1pt,linecolor=gray]{}(-1,3.13)(3,3.13)
	\psline[linestyle=dotted,dotsep=1pt,linecolor=gray]{}(-1,0.87)(3,0.87)
	\psline[linestyle=dotted,dotsep=1pt,linecolor=gray]{}(1.8,2.13)(2.9,2.13)
	\psline[linestyle=dotted,dotsep=1pt,linecolor=gray]{}(2.2,1.87)(2.9,1.87)
	\psline[]{<->}(1.8,-0.6)(2.2,-0.6)	
	
	\psline[]{->}(0,0)(4.5,0)
	\psline[]{->}(0,0)(0,3.5)
	\rput[l](4.6,0){$t$}
	\rput[b](0,3.6){Power}
	
	\psline[linestyle=dashed]{-}(0.5,3)(3.5,1)
	\qdisk(0.5,3){2pt}
	\qdisk(3.5,1){2pt}
	\psline[](1.8,3.13)(2.2,0.87)
	\qdisk(1.8,3.13){2pt}
	\qdisk(2.2,0.87){2pt}
	
	\psline[]{<->}(2.6,0.87)(2.6,1.87)\rput[l](2.7,1.2){$\gamma$}
	\psline[]{<->}(2.6,2.13)(2.6,3.13)\rput[l](2.7,2.6){$\gamma$}
	\psline[]{-}(2.8,1.67)(2.8,2.33)\rput[l](2.95,2.0){${(p\rf_{s}-p\rf_{s-1})\frac{T\ctrl}{T\sys}}$}
	\psline[]{->}(2.8,1.67)(2.8,1.87)
	\psline[]{->}(2.8,2.33)(2.8,2.13)
	
	\rput[r](-0.2,2.9){$p\rf_{s-1}$}
	\rput[r](-0.2,1.1){$p\rf_{s}$}
	\rput[B](0.5,-0.5){$(s-1)T\sys$}
	\rput[B](3.5,-0.5){${sT\sys}$}
	\rput[B](2.0,-0.5){$T\ctrl$}
	\rput[r](-1.0,3.1){$p\opt((l-1)T\ctrl)$}
	\rput[r](-1,0.87){${p\opt(lT\ctrl)}$}
	
\end{pspicture*}
\caption{The target power trajectory $p\opt(t)$ (solid line) requires maximum negative ramp rates if ${w_{l-1}=-w_{l}=1}$.}
\label{fig:ramps}
\end{figure}

\subsection{State constraints}
Many types of flexible energy resources owe their flexibility to an energy buffer with limited capacity. Examples of such systems are pumped hydro-power plants, batteries, heating and cooling systems, capacitors, and flywheels. Here we consider systems whose energy buffer level ${x(t)\in\R{}}$ is governed by the linear time-invariant dynamics
\begin{equation}
	\partial x(t,w)/\partial t=ax(t,w)+bu_s+cp\opt(t,w),
	\label{eq:stateDynamics}
\end{equation}
with ${t\in\intTsys}$, ${s\in\set{N}\sys}$. Exogenous (uncontrollable) inputs, such as the weather conditions in the case of heating systems, or trips in the case of electric-vehicle batteries, are summarized by ${u:=[u_1,\dots,u_{N\sys}]\trans}$ and are assumed to be known. Because the bidding problem has to be solved ahead of time, the initial state $x(0)$ is not known precisely, but is known to lie in the interval ${[x_0^{\min},x_0^{\max}]}$. The dynamics \eqref{eq:stateDynamics} are characterized by the parameters ${\{a,b,c\}}$, where ${a\leq 0}$ can be interpreted as the self-dissipation rate of the energy buffer, and $b,c$ determine the conversion efficiencies of exogenous inputs and electric energy into buffered energy, respectively. Our sign convention implies ${c\geq 0}$. The behavior of an ideal battery, for instance, can be modeled by setting ${a=b=0}$ and ${c=1}$.

We consider the piece-wise constant state constraints
\begin{equation}
	\un{x}_{s}\leq x(t,w) \leq\bar{x}_{s},\ \forall w\in\set{W}^{N\ctrl},
	\label{eq:stateConstrCt}
\end{equation}
for all ${t\in\intTsys}$, ${s=\set{N}\sys}$, where ${\un{x},\bar{x}\in\R{N\sys}}$ denote the physical limits of the system state. The initial state is assumed to be feasible, \ie, ${\un{x}_1\leq x^{\min}_0\leq x^{\max}_0\leq\bar{x}_1}$. 

The exact evolution of the state \eqref{eq:stateDynamics} is unknown because it depends on the uncertain activation ${w\in\set{W}^{N\ctrl}}$. 
For our case of piece-wise affine power references, the state $x(t,w)$ can evolve non-monotonically during a discretization interval $\intTsys$. 
To ensure that \eqref{eq:stateConstrCt} holds would require bounding the extreme values of $x(t,w)$ for ${t\in\intTsys}$. This, however, would result in non-convex constraints. Instead, we approximate the right-hand side of \eqref{eq:stateDynamics} by affine functions of time and define for every interval $\intTsys$, ${s\in\set{N}\sys}$, the auxiliary dynamics
\begin{align}
	\dot{y}\ind{s}(\tau,\tilde{w}) :=&\,a\bar{x}_s+bu_s+c(p\rf_{s-1}(\tilde{w})-\gamma)\nonumber\\
	&+[\min\{0,a\bar{x}_s+bu_s+c(p\rf_{s}(\tilde{w})-\gamma)\}\label{eq:defAuxDyn1}\\
	&-(a\bar{x}_s+bu_s+c(p\rf_{s-1}(\tilde{w})-\gamma))] \tau/T\sys,\nonumber\\
	\dot{z}\ind{s}(\tau,\tilde{w}) :=&\,a\un{x}_s+bu_s+c(p\rf_{s-1}(\tilde{w})+\gamma)\nonumber\\
	&+[\max\{0,a\un{x}_s+bu_s+c(p\rf_{s}(\tilde{w})+\gamma)\}\label{eq:defAuxDyn2}\\
	&-(a\un{x}_s+bu_s+c(p\rf_{s-1}(\tilde{w})+\gamma))]\tau/T\sys,\nonumber
\end{align}
with ${\tau\in [0,T\sys]}$, ${y\ind{1}(0)=x_0^{\min}}$, ${z\ind{1}(0)=x_0^{\max}}$, and ${y\ind{s}(0,\tilde{w})=z\ind{s}(0,\tilde{w})=x((s-1)T\sys,\tilde{w})}$. For brevity we use $\dot{y}\ind{s}(\tau,\tilde{w})$ to denote $\partial y\ind{s}(\tau,\tilde{w})/\partial\tau$

\begin{prop}
	If for all ${s\in\set{N}\sys}$ and ${\tilde{w}\in\set{W}^{N\sys}}$
	\begin{equation}
		y\ind{s}(\tau,\tilde{w}) \geq \un{x}_s\ts{ and }z\ind{s}(\tau,\tilde{w}) \leq \bar{x}_s	,\,\tau\in [0,T\sys],
		\label{eq:yzConstrCt}
	\end{equation}	
	then the state constraints \eqref{eq:stateConstrCt} are satisfied for all ${t\in\intTsys}$.
\end{prop}
\begin{proof}
	Consider the time interval $\intTsys$ and divide it into ${\kappa\in\set{N}^+}$ subintervals, each of duration ${\delta:=T\sys/\kappa}$. For ${\tau=0}$, the constraints \eqref{eq:yzConstrCt} are equivalent to ${\un{x}_s\leq x((s-1)T\sys,w)\leq \bar{x}_s}$, ${\forall\tilde{w}\in\set{W}^{N\sys}}$. By definitions \eqref{eq:defAuxDyn1}--\eqref{eq:defAuxDyn2}, it follows that ${\dot{y}\ind{s}(0,\tilde{w})\leq \dot{x}((s-1)T\sys,w)\leq \dot{z}\ind{s}(0,\tilde{w}),\,\forall \tilde{w}\in\set{W}^{N\sys}}$.
Approximate
\begin{align}
	x((s-1)T\sys+\delta,w) &\approx x((s-1)T\sys,w) + \delta\dot{x}((s-1)T\sys,w),\\	
	y\ind{s}(\delta,\tilde{w}) &\approx x((s-1)T\sys,w)+\delta\dot{y}\ind{s}(0,\tilde{w}),\\	
	z\ind{s}(\delta,\tilde{w}) &\approx x((s-1)T\sys,w)+\delta\dot{z}\ind{s}(0,\tilde{w}),
\end{align}
and note that ${y\ind{s}(\delta,\tilde{w})\leq x((s-1)T\sys+\delta,w)\leq z\ind{s}(\delta,\tilde{w})}$, which, via \eqref{eq:yzConstrCt}, implies that ${\un{x}_s\leq x((s-1)T\sys+\delta,w)\leq\bar{x}_s}$. The same line of arguments holds for all subsequent subintervals. Consequently, in the limit case ${\delta\rightarrow 0}$, we have that ${\un{x}_s\leq x(t,w) \leq\bar{x}_s}$ ${\forall t\in\intTsys}$, ${s\in\set{N}\sys}$, ${w\in\set{W}^{N\ctrl}}$.
\end{proof}

Note that $\dot{y}\ind{s}(\tau,\tilde{w}),\,\dot{z}\ind{s}(\tau,\tilde{w})$ are affine functions of $\tau$ that can be integrated explicitly. In addition, the $\min\{\cdot\}$ and the $\max\{\cdot\}$ in the definitions \eqref{eq:defAuxDyn1}--\eqref{eq:defAuxDyn2} ensure that the minimum value of $y\ind{s}(\tau,\tilde{w})$ and the maximum value of $z\ind{s}(\tau,\tilde{w})$ over ${\tau\in [0,T\sys]}$ are reached at the edges of the interval, \ie, at ${\tau\in\{0,T\sys\}}$. Thus, \eqref{eq:yzConstrCt} holds if for all ${s\in\set{N}\sys}$ and ${\tilde{w}\in\set{W}^{N\sys}}$
\begin{equation}
		y\ind{s}(\tau,\tilde{w})\geq\un{x}_s,\, z\ind{s}(\tau,\tilde{w}) \leq \bar{x}_s,\,\tau\in \{0,T\sys\},
\end{equation}	
which is equivalent to requiring that for all ${s\in\set{N}\sys}$,  ${w\in\set{W}^{N\ctrl}}$,  and ${\tilde{w}\in\set{W}^{N\sys}}$
\begin{align}
	\un{x}_s\leq x((s-1)T\sys,w) \leq \bar{x}_s,\label{eq:xsConstr}\\
	y\ind{s}(T\sys,\tilde{w})\geq \un{x}_s,\label{eq:ysConstr}\\
	z\ind{s}(T\sys,\tilde{w})\leq \bar{x}_s,\label{eq:zsConstr}
\end{align}
Thus, \eqref{eq:xsConstr}-\eqref{eq:zsConstr} are sufficient conditions for the state constraints \eqref{eq:stateConstrCt}.
%
The robust counterparts of \eqref{eq:xsConstr} are derived in Appendix \ref{a:stateDynamics}. The robust counterparts of \eqref{eq:ysConstr} and \eqref{eq:zsConstr} can be found similarly.

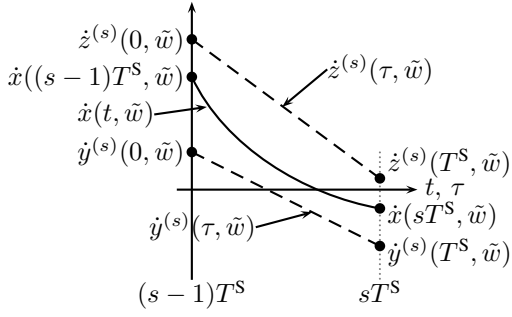
\begin{figure}[tb]
\setlength{\belowcaptionskip}{-5pt}
\centering
\psset{unit=1}
\begin{pspicture*}(-3,-1.6)(5.0,2.5)
	
	\psline[linestyle=dotted,dotsep=1pt,linecolor=gray]{}(2.5,-1.2)(2.5,0.5)
	
	
	
	\psline[]{->}(-0.2,0)(3,0)
	\psline[]{->}(0,-1.2)(0,2.5)
	\rput[l](3.1,0){$t,\,\tau$}
	
	\psline[linestyle=dashed]{}(0,2)(2.5,0.15) 
	\qdisk(0,2){2pt}
	\qdisk(2.5,0.15){2pt}
	\psline[linestyle=dashed]{}(0,0.5)(2.5,-0.75) 
	\qdisk(0,0.5){2pt}
	\qdisk(2.5,-0.75){2pt}
	
	\pnode(0,1.5){A} 
	\pnode(2.5,-0.25){B}
	\qdisk(0,1.5){2pt}
	\qdisk(2.5,-0.25){2pt}
	\nccurve[angleA=-65,angleB=170]{A}{B}
	
	\rput[B](0.0,-1.5){${(s-1)T\sys}$}
	\rput[B](2.5,-1.5){${sT\sys}$}
	
%
%
	\rput[r](-0.1,2){$\dot{z}\ind{s}(0,\tilde{w})$}
	\rput[lB](2.6,0.25){${\dot{z}\ind{s}(T\sys,\tilde{w})}$}
	
	\rput[r](-0.1,0.5){$\dot{y}\ind{s}(0,\tilde{w})$}
	\rput[lB](2.6,-0.95){$\dot{y}\ind{s}(T\sys,\tilde{w})$}
	
	\rput[r](-0.1,1.5){${\dot{x}((s-1)T\sys,\tilde{w})}$}
	\rput[lB](2.6,-0.45){$\dot{x}(s T\sys,\tilde{w})$}
	
	\psline[]{->}(-0.5,1)(0.2,1.1)\rput[r](-0.5,1){$\dot{x}(t,\tilde{w})$}
	\psline[]{<-}(1.2,1.1)(1.8,1.6)\rput[l](1.8,1.6){$\dot{z}\ind{s}(\tau,\tilde{w})$}
	\psline[]{->}(0.85,-0.5)(1.6,-0.3)\rput[r](0.85,-0.5){$\dot{y}\ind{s}(\tau,\tilde{w})$}
	
\end{pspicture*}
\caption{Illustration of ${\dot{y}\ind{s}(\tau,\tilde{w})\leq\dot{x}(t)\leq\dot{z}\ind{s}(\tau,\tilde{w})}$ for ${t\in\intTsys}$, $\tau\in [0,T\sys]$.}
\label{fig:stateApprox}
\end{figure}

\subsection{Description of flexibility}
The power, ramp-rate, and state constraints \eqref{eq:powerConstrCt}, \eqref{eq:rampConstrCt}, and \eqref{eq:stateConstrCt}, respectively, are satisfied if the offered SFR capacity $\gamma$ together with the energy trading policies ${(Q\da,q\da)}$, ${(Q\id,q\id)}$ meet the constraints \eqref{eq:powerConstrRobust}, \eqref{eq:rampConstrRobust}, and the robust counterparts of \eqref{eq:xsConstr}-\eqref{eq:zsConstr}. All these constraints are linear in the decision variables ${\zeta:=\{Q\da,q\da,Q\id,q\id,\gamma\}}$ and thus define a convex polytope ${\set{P}(\Phi)}$ parameterized by the system constraints parameters ${\Phi:=\{\un{p},\bar{p},\un{r},\bar{r},\un{x},\bar{x},x_0^{\min},x_0^{\max}\}}$. This polytope can serve as a concise and convenient description of the flexibility the system offers with regard to trading energy in the day-ahead and intra-day markets and offering regulation power on the SFR market.

\section{Solving the Bidding Problem}
\label{s:scheduling}

\subsection{Formulation of the bidding problem}

The energy and reserve bidding problem consists of finding the energy trading policies \eqref{eq:defControllers} and the reserve capacity $\gamma$ that maximize the expected profit while satisfying all system constraints. 
%
%
The profit
\begin{align}
J(\zeta,w) :=&\ R\rsv(\gamma)+R\reg(\gamma,w)-C\da(Q\da,q\da,w)\nonumber\\
&\ -C\id(Q\id,q\id,w),
\label{eq:defObjective}
\end{align}
is uncertain because of its dependency on the SFR activation signal and the market clearing prices. Thus, we maximize the expectation of \eqref{eq:defObjective} over the uncertain parameters ${\Psi:=\{w,c\da,c\id,c\rsv,c^{\ts{up}},c^{\ts{dn}}\}}$, and write the energy and reserve bidding problem as the linear program

\begin{equation}
	\max\limits_{\zeta} \mathbf{E}_{\Psi}[ J(\zeta,w)]\ \ts{ s.t. } \zeta\in\set{P}(\Phi).
	\label{eq:schedulingOpti}
\end{equation}
The number of decision variables and constraints of \eqref{eq:schedulingOpti} is determined by various factors, such as the planning horizon, the different market timescales, the system discretization time, and also the trading policies \eqref{eq:defControllers} used.

\subsection{Computational study}
\label{ss:simulationStudy}
The results of the economic bidding problem \eqref{eq:schedulingOpti} do not only depend on the market characteristics and technical constraints of the system, but are also influenced by market prices. An economic analysis of \eqref{eq:schedulingOpti} can be found in \cite{Qureshi2016}. To investigate the maximum available SFR capacity independent of economic considerations, we solve \eqref{eq:schedulingOpti} with objective ${\max_{\zeta}\gamma}$ and investigate the effects that different market time scales, lead times, and trading policies have on the maximum available SFR capacity $\gamma\opt$.  
To facilitate the interpretation of the results, we consider an ideal battery modeled according to \eqref{eq:stateDynamics} with parameters ${a=0}$, ${b=c=1}$, and no exogenous input, \ie, ${u=0}$. The physical constraints of the battery are similar to those of a \emph{Tesla Powerwall 2} battery \cite{Powerwall2016}: ${-\un{p}_s=\bar{p}_s=5}$ kW, ${\un{x}_s=0}$ kWh, ${\bar{x}_s=15}$ kWh, ${s=1,\dots,N\sys}$, and $x_0=7.5$ kWh. We use ${T\sys=5}$ min, ${T\ctrl=1}$ s, ${N\da_{\ts{lb}}=13}$ h, and ${\hor=T\rsv}$.
Table \ref{tab:capResults} provides the market settings and trading policy parameters considered and the results obtained. 
%
The maximum available SFR capacity $\gamma\opt$ and the minimum ramp rates $\bar{r}\opt$ required to offer this service are provided as quantities relative to the rated power $\bar{p}$.

As a reference case, the first setting described in \tab\ref{tab:capResults} considers the case where the power reference cannot be adjusted via day-ahead or intra-day energy trades; it remains fixed regardless of potential SFR activation. This setting results in the least amount of SFR capacity, which, for the simple system at hand, can be computed analytically by dividing the available up- and down energy buffer capacity by the duration of the planning horizon: ${\gamma\opt=}$ 7.5 kWh / (7$\cdot$24) h $\approx$ 0.0446 kW, which corresponds to 0.89\% of $\bar{p}$. 
%
More SFR capacity can be made available by adjusting the power reference via day-ahead energy trades depending on past activation. The available capacity $\gamma\opt$ increases as trading decisions are based on more data on past activation, see settings 2--6, where the value of $N\da_{\ts{lb}}$ is gradually increased. Considering more than 24 h of past activation data, \ie, $N\da_{\ts{lb}}>24$, does not yield more SFR capacity because the activation that occurred more than 24 h in the past has already been compensated for in the previous day-ahead market. 

Whereas day-ahead adjustments of the power reference yield only minor increases of SFR capacity over the reference setting, intra-day adjustments make more than 50\% of the rated power available for SFR, see settings 7--9. Reducing the intra-day market lead time $T\id_{ld}$ from 1 h to 1/4 h, however, has only minor effects on $\gamma\opt$, see settings 11--13. Reducing the planning horizon from 7 d to 1 d results in an increase of $\gamma\opt$ from 50.26\% to 51.87\%, see settings 7 and 11. 
The results also illustrate the ramp-rate constraints \eqref{eq:rampConstr}. Consider setting 13, for instance: Tracking SFR activation of size ${\gamma\opt=52.63}$ \% of $\bar{p}$ around a constant $p\rf$ would require ${\bar{r}=2\cdot 52.63/T\ctrl=105.26}$ (\% of $\bar{p}$)/s, whereas in reality the slightly larger amount ${\bar{r}=105.42}$ (\% of $\bar{p}$)/s is needed.

The results highlight the importance of the intra-day market which operates at faster time scales than the day-ahead market. It allows energy-constrained systems such as batteries to compensate for the regulation energy provided. The introduction of energy markets operating at even shorter time intervals, \eg 5 min, could further increase the amount of available reserves. Additional reserves can be made available by reducing the length of the tendering period $T\rsv$ and by relaxing the requirement that a constant amount of SFR must be offered.

The bidding problem \eqref{eq:schedulingOpti} has been solved with \textsc{IBM ILOG}\textregistered\footnote{\label{fnote}ILOG and CPLEX are trademarks of International Business Machines Corp., registered in many jurisdictions worldwide. Other product and service names might be trademarks of IBM or other companies.} \textsc{CPLEX}\textregistered\footnotemark[3]\, v12.6 on a system featuring an Intel\textregistered\, Xeon\textregistered\, 12-core CPU @2.4 GHz and 20 GB of RAM. The solving times are provided in the rightmost column of \tab\ref{tab:capResults}. For fixed discretization time $T\sys$, the problem size grows quadratically with the planning horizon length $\hor$. An efficient means to reduce the problem size is to reduce the number of free variables in the energy trading policies \eqref{eq:defControllers} via the parameters $N_{lb}\da$ and $N_{lb}\id$, \cf Section \ref{ss:adpatingSchedules}. 

\begin{table}[t]
\setlength{\abovecaptionskip}{-5pt}
\caption{Maximum SFR capacity and minimum required ramp rates for different market  settings and trading policies.}
\label{tab:capResults}
\begin{center}
\begin{tabular}{|c|c|c|c|c|c|c|c|c|c|} \hline 
\# & ${T\rsv}$ &  $T\id_{ld}$ & $N_{lb}\da$ & $N_{lb}\id$ & $\gamma\opt$ & $\bar{r}\opt$ & Sol. time\\ 
 & [d] & [h] &  & & [\% $\bar{p}$] & [(\% $\bar{p}$)/s] & \\\hline
1 & 7 & 1 & 0  & 0 & 0.89 & 1.78 & 8 min\\ 
2 & 7 & 1 & 1  & 0 & 0.93 & 1.87 & 9 min\\ 
3 & 7 & 1 & 2  & 0 & 0.96 & 1.94 & 12 min\\ 
4 & 7 & 1 & 6 & 0 & 1.14 & 2.29 & 40 min\\ 
5 & 7 & 1 & 12 & 0 & 1.55 & 3.13 & 100 min\\
6 & 7 & 1 & $\geq$ 24 & 0 & 4.05 & 8.25 & 270 min\\ \hline
7 & 7 & 1 & 0 & $\geq$ 1 & 50.26 & 100.69 & 30 min\\ 
8 & 7 & 1/2 & 0 & $\geq$ 1 & 50.34 & 100.84 & 30 min\\ 
9 & 7 & 1/4 & 0 & $\geq$ 1 & 50.37 & 100.91 & 30 min\\ \hline
10 & 1 & 1 & 0 & 0 & 6.25 & 12.50 & 5 s\\ 
11 & 1 & 1 & 0 & $\geq$ 1 & 51.87 & 103.90 & 25 s\\ 
12 & 1 & 1/2 & 0 & $\geq$ 1 & 52.38 & 104.92 & 25 s\\ 
13 & 1 & 1/4 & 0 & $\geq$ 1 & 52.63 & 105.42 & 25 s\\ 
\hline
\end{tabular}
\end{center}
\end{table}

\section{Conclusion}
\label{s:conclusion}
The provision of accurate and reliable SFR services is important to compensate for imbalances in the power grid. We showed how to implement energy schedules in a way that is compliant with ENTSO-E regulations and allows one to explicitly consider limitations of the power-ramp rates of the system providing SFR. Based on piece-wise affine and continuous power trajectories, we derived a reformulation of the energy and reserve bidding problem as a robust linear program that accurately models the different market timescales and lead times. Our approach is versatile and applicable in various market settings. Computations show how adjusting the power reference via day-ahead and intra-day energy markets significantly increases the amount of available SFR capacity. 

%
\IEEEpeerreviewmaketitle

\appendices

\section{Discrete state dynamics}
\label{a:stateDynamics}

Consider the linear time-invariant state dynamics \eqref{eq:stateDynamics}: $${\partial x(t,w)/\partial t = ax(t,w)+bu_s+cp\opt(t,w)},$$ ${t\in\intTsys}$, ${s\in\set{N}\sys}$, with ${x_0:=x(0)\in [x_0^{\min},x_0^{\max}]}$ and $p\opt(t,w)$ according to \eqref{eq:defTargetPower}. Standard integration techniques can be used to express the state ${x_s(w):=x(sT\sys,w)}$ as
\begin{align}
	x_{s}(w) =&\, fx_{s-1}(w)+gu_s+ h_1p\rf_{s-1}(\tilde{w})+h_2p\rf_{s}(\tilde{w}) + \gamma v_s,
	\label{eq:discreteDynamics}
\end{align}
where we have used
\begin{align}
	f &:= e^{aT\sys},\hspace*{21mm} g :=b\int_0^{T\sys}e^{a\tau}d\tau,\\
	h_1 &:=\frac{c}{T\sys}\int_0^{T\sys}e^{a\tau}\tau d\tau,\ \ 
	h_2 :=\frac{c}{T\sys}\int_0^{T\sys}e^{a\tau}(T\sys-\tau) d\tau,\\
	v_s &:=c\int_0^{T\sys}e^{a\tau} w(sT\sys-\tau)d\tau.
\end{align}
Note that the power reference is determined by the averaged activation $\tilde{w}$, whereas the term $v_s$ depends on the unfiltered activation signal $w(t)$. Because $\tilde{w}$ and $v$ are computed from $w$, we write $x(w)$ only.
The discrete-time state evolution can be written in vector form as
\begin{equation}
	x(w) = Fx_0+Gu+Hp\rf(\tilde{w})+\gamma K v,
	\label{eq:discreteDynLD}
\end{equation} 
with ${F\in\R{N\sys}}$, ${G,K\in\R{N\sys\times N\sys}}$, and ${H\in\R{N\sys\times (N\sys+1)}}$.

By plugging \eqref{eq:defPowerRef} into \eqref{eq:discreteDynLD}, we have
\begin{equation}
	x_s(w) = \sum\limits_{i=1}^{s}\{ Q_{s,i}\tilde{w}_i + \gamma K_{s,i} v_i\} + \kappa_s,
\end{equation}
where we have used 
\begin{align}
	Q &:= HR(MQ\da A\da + Q\id A\id)\\
	\kappa &:= Fx_0+Gu+HR(Mq\da +q\id)
\end{align}
for brevity. Using the definitions of $\tilde{w}$ and $v$, we can write
\begin{equation}
	x_s(w) = \sum\limits_{i=1}^{s}\left\{\int_0^{T\sys}\left(\frac{Q_{s,i}}{T\sys}  + c\gamma e^{a\tau} K_{s,i} \right)w(\tau)d\tau\right\} + \kappa_s.
	\label{eq:xsSum}
\end{equation}
Note that in case of self-dissipation, \ie, ${a<0}$, the integrand is time-dependent, making it difficult to evaluate the integral analytically. To circumvent this, we substitute $e^{a\tau}$ by the constant $e^{a\hat{\tau}}$ which is chosen such that for all ${\tau\in [0,T\sys]}$
\begin{equation}
	\abs{e^{a\tau}-e^{a\hat{\tau}}}\leq \epsilon:=(1-e^{a T\sys})/2.
\end{equation}
Because 
\begin{multline}
	\max\limits_{w\in\mathcal{W}}\int_0^{T\sys}c\gamma K_{s,i}(e^{a\tau}-e^{a\hat{\tau}})w(\tau)d\tau\\
	= \int_0^{T\sys}c\gamma K_{s,i}\abs{e^{a\tau}-e^{a\hat{\tau}}} d\tau\leq c\gamma\epsilon T\sys K_{s,i},
\end{multline}
the integral in \eqref{eq:xsSum} is bounded from above:
\begin{align}
& \int_0^{T\sys}\left(\frac{Q_{s,i}}{T\sys}  + c\gamma e^{a\tau} K_{s,i} \right)w(\tau)d\tau \nonumber\\
 \leq& \int_0^{T\sys}\left(\frac{Q_{s,i}}{T\sys}  + c\gamma e^{a\hat{\tau}} K_{s,i} \right)w(\tau)d\tau + c\gamma\epsilon T\sys K_{s,i}\nonumber\\
  \leq&\, \abs{Q_{s,i} + c\gamma e^{a\hat{\tau}}T\sys K_{s,i}} + c\gamma\epsilon T\sys K_{s,i}\ \ \forall w\in\set{W}^{N\ctrl}.
  \label{eq:aux3}
\end{align}
Bounds on $x_s(w)$ can be found by substituting the summands in \eqref{eq:xsSum} with \eqref{eq:aux3}: For all ${w\in\set{W}^{N\ctrl}}$ and ${s\in\set{N}\sys}$
\begin{align}
\begin{split}
	x_s(w) &\leq \abs{Q_{s} + c\gamma e^{a\hat{\tau}}T\sys K_{s}}\mathbf{1} + c\gamma\epsilon T\sys K_{s}\mathbf{1} + \kappa_s, \\ 
	x_s(w) &\geq -\abs{Q_{s} + c\gamma e^{a\hat{\tau}}T\sys K_{s}}\mathbf{1} - c\gamma\epsilon T\sys K_{s}\mathbf{1} + \kappa_s. \label{eq:defXsMaxMin}
\end{split}
\end{align}
%
%
Inequalities \eqref{eq:defXsMaxMin} are the robust counterparts of \eqref{eq:xsConstr}.



\ifCLASSOPTIONcaptionsoff
  \newpage
\fi

\bibliographystyle{ieeetr}
\bibliography{mylibrary}



%
%
%

\vspace*{-3\baselineskip}
\vfill
\begin{IEEEbiographynophoto}
{Fabian L. M\"uller} received the B.Sc. degree in mechanical engineering and the M.Sc. degree in robotics, systems and control from the Swiss Federal Institute of Technology (ETH) Zurich, Switzerland, in 2009 and 2012, respectively. 

He was recently affiliated with the Institute for Dynamic Systems and Control at ETH Zurich. 
 
Mr. M\"uller is currently a doctoral student with the Automatic Control Laboratory at ETH Zurich and with IBM Research Zurich. His research interests include technologies for smart power grids, control theory, and optimization.
\end{IEEEbiographynophoto}
\vspace*{-2.55\baselineskip}
\begin{IEEEbiographynophoto}
{Stefan Woerner} received the M.Sc degree in applied mathematics and the doctoral degree from the Department of Management, Technology and Economics, Swiss Federal Institute of Technology (ETH) Zurich, Switzerland, in 2010 and 2013, respectively.

He is currently a research staff member at IBM Research in Zurich where he leads multiple projects on supply chain management and forecasting.

Dr. Woerner's main research interests are in approximate dynamic programming, robust optimization, simulation optimization, and applications of these techniques to problems in supply chain management and energy systems.
\end{IEEEbiographynophoto}
\vspace*{-2.55\baselineskip}
\begin{IEEEbiographynophoto}
{John Lygeros} (F'11) received the B.Eng. degree in electrical engineering and the M.Sc. degree in systems control from the Imperial College of Science Technology and Medicine, London, U.K., in 1990 and 1991, respectively, and the Ph.D. degree from the Electrical Engineering and Computer Sciences Department, University of California at Berkeley (UC Berkeley), Berkeley, CA, USA, in 1996.

After Post-Doctoral Researcher appointments at MIT and UC Berkeley in 2000, he joined the Department of Engineering, University of Cambridge, Cambridge, U.K., as a Lecturer and Churchill College, Cambridge, as a fellow. From 2003 to 2006, he was an Assistant Professor with the Department of Electrical and Computer Engineering, University of Patras, Patras, Greece. In 2006, he joined the Automatic Control Laboratory, ETH Zurich, Zürich, Switzerland, as an Associate Professor, where he has been a Full Professor since 2010. Since 2009, he has been serving as the Head of the Automatic Control Laboratory, and since 2015, he has been serving as the Head of the Department of Information Technology and Electrical Engineering. His current research interests include modeling, analysis, and control of hierarchical, hybrid, and stochastic systems, with applications to biochemical networks, automated highway systems, air traffic management, power grids, and camera networks. 

Dr. Lygeros is a member of the IET and the Technical Chamber of Greece. Since 2013, he has been serving as the Treasurer of the International Federation of Automatic Control.
\end{IEEEbiographynophoto}

%









\end{document}